\documentclass[12pt]{amsart}

\usepackage{amsmath,amssymb,amsthm,bm,epsfig,color}

\theoremstyle{plain}
\newtheorem{theorem}{Theorem}[section]

\newtheorem{proposition}[theorem]{Proposition}

\newtheorem{lemma}[theorem]{Lemma}

\newtheorem{remark}[theorem]{Remark}

\numberwithin{equation}{section}

\newcommand{\bbR}{{\mathbb R}}

\newcommand{\cL}{{\mathcal L}}
\newcommand{\cD}{{\mathcal D}}
\newcommand{\cN}{{\mathcal N}}
\newcommand{\cH}{{\mathcal H}}
\newcommand{\cE}{{\mathcal E}}

\newcommand{\cC}{{\mathcal C}}
\newcommand{\cB}{{\mathcal B}}

\renewcommand{\Im}{{\mathrm{Im}}}

\newcommand\ulambda{{\underline{\lambda}}}
\newcommand\Xom{{\mathcal L}^2}

\newcommand\tepsilon{{\tilde{\epsilon}}}

\begin{document}

\title[]{Arrest of blowup for the 3-D semi-relativistic Schr\"odinger-Poisson system with pseudo-relativistic diffusion}


\author{W. Abou Salem, T. Chen and V. Vougalter}


\address{Department of Mathematics and Statistics, University of Saskatchewan, Saskatoon S7N 5E6, Canada \\ E-mail: w.abousalem@usask.ca}

\address{Department of Mathematics, University of Texas at Austin, Austin, TX, 78712, USA \\ E-mail: tc@math.utexas.edu}
 
\address{Department of Mathematics and Applied Mathematics, University of Cape Town, Private Bag, Rondebosch 7701, South Africa \\E-mail: vitali@math.toronto.edu}

\maketitle


\begin{abstract} We show global well-posedness in energy norm of the semi-relativistic Schr\"odinger-Poisson system of equations with attractive Coulomb interaction in ${\mathbb R}^3$ in the presence of pseudo-relativistic diffusion. We also discuss sufficient conditions to have well-posedness in $\cL^2.$ In the absence of dissipation, we show that the solution corresponding to an initial condition with negative energy blows up in finite time, which is as expected, since the nonlinearity is critical. 
\end{abstract}

\vspace{1cm}


\section{Introduction and heuristic discussion}

The main purpose of this work is to show for solutions 
to the semi-relativistic Schr\"odinger-Poisson system with 
attractive Coulomb interaction that
an arbitrarily small pseudo-relativistic diffusion term arrests blow up. 
This system describes the mean-field dynamics of interacting bosons in a {\it mixed} state, see \cite{A10}. For a discussion about the mathematical properties of the pseudo-relativistic diffusion term, see \cite{CMS90, R02}, and \cite{BMN10} for a physical discussion. The special case when the bosons are in a pure state yields the semi-relativistic Hartree equation, whose well-posedness and blow up has been studied in \cite{L07, FL07, CO07}.

For attractive Coulomb interaction, and in the absence of dissipation, the global well-posedness in energy space for small initial data in $L^2$-norm follows from Theorem 1.1 in \cite{ACV13}, see also \cite{ACV12, ACV13-2}. Due to the criticality of the attractive nonlinearity, the solution is expected to blow up in finite time in energy space. Here, we show that this is indeed the case for initial data with negative energy and in the absence of dissipation; however, pseudo-relativistic diffusion arrests the blow up. 

We note that there has been a substantial amount of work on the effects of dissipation on blow up of solutions for the nonrelativistic nonlinear Schr\"odinger (NLS) equation in recent years. In \cite{PSS05}, it is shown that when linear or nonlinear
dissipative terms are incorporated in the focusing cubic NLS equation in two spatial dimensions, the equation becomes globally well-posed in $H^{1}({\mathbb R}^{2})$. The effect of dissipation is also studied numerically in \cite{FK12} for the critical and supercritical NLS equation with the nonlinear  dissipation. In paper \cite{AS10} the authors show the global  well-posedness for the cubic NLS with nonlinear damping when
the external quadratic confining potential is present. Our analysis can be viewed as a generalization of some of these results to a semi-relativistic system with infnitely coupled NLS equations. The analysis below can also be applied to study arrest of collapse for semi-relativistic Hartree-Fock, see \cite{HLLS10} where blowup in absence of dissipation is studied.

The organization of this note is as follows. In section \ref{sec:MainResults}, we describe our model and state the main result.  In section \ref{sec:Dissipation}, we prove local and then global well-posedness of the Schr\"odinger-Poisson  system of equations in the presence of dissipation. Finally, in section \ref{sec:blow up}, we prove that the solution blows up in energy space in the absence of dissipation and for suitable initial conditions. 

\subsection{Notation}
\begin{itemize}

\item $A\lesssim B$ means that there exists a positive constant $C$ such that $A\le C\; B.$ Moreover, $A\sim B$ means $A\lesssim B$ and $B\lesssim A.$
\item $L^p$ stands for the standard Lebesgue space. Furthermore, $L_I^pB = L^p(I;B).$ $\langle \cdot, \cdot \rangle_{L^2}$ denotes the $L^2(\bbR^n)$ inner product, which is given by 
$$\langle u,v\rangle_{L^2} = \int_{{\mathbb R}^3} \overline{u}(x)v(x)dx.$$
 We will often use the abbreviated notation $L^p_T$ for $L^p_{[0,T]}$, 
in the situation where $[0,T]$ denotes a time interval.

\item $l^1 = \{\{a_l\}_{l\in {\mathbb N}} | \ \ \sum_{l\ge 1} |a_l| <\infty\}.$
\item   $W^{s,p} = (-\Delta+1)^{-\frac{s}{2}}L^p,$ the standard (complex) Sobolev space. When $p=2$, $W^{s,2}= H^s$. 
$\dot{H}^s$ denotes the homogeneous Sobolev space with norm 
$\|\psi\|_{\dot{H}^s} = (\langle \psi, (-\Delta)^s \psi\rangle_{L^2})^{\frac{1}{2}}.$
\item For fixed $\ulambda\in l^1, \ \ \lambda_k> 0,$ and for sequences of functions
$\Phi:=\{\phi_{k}\}_{k\in {\mathbb N}}$ and 
$\Psi:=\{\psi_{k}\}_{k\in {\mathbb N}},$ we define the inner product
$$
\langle\Phi,\Psi\rangle_{\Xom}:=\sum_{k\ge 1}\lambda_{k}\langle\phi_{k}, 
\psi_{k}\rangle_{L^{2}},
$$ 
which induces the norm 
$$
\|\Phi\|_{\Xom}=(\sum_{k\ge 1}\lambda_{k}{\|\phi_{k}\|
_{L^{2}}^{2}})^{\frac{1}{2}}.
$$
The corresponding Hilbert space is $\Xom.$
\item For fixed $\ulambda\in l^1, \ \ \lambda_k> 0,$ $$\cH^s = \{ \Psi=\{\psi_k\}_{k\in {\mathbb N}} | \ \ \psi_k\in H^s, \ \ \sum_{k\ge 1} \lambda_k \|\psi_k\|_{H^s}^2<\infty \}$$ is a Banach space with norm
$\|\Psi\|_{\cH^s} = (\sum_{k\ge 1}\lambda_k \|\psi_k\|_{H^s}^2)^{\frac{1}{2}}.$ In particular, $\cH^0 = \cL^2.$

\item For fixed $\ulambda\in l^1, \ \ \lambda_k> 0,$ $$\dot{\cH}^s = \{ \Psi=\{\psi_k\}_{k\in {\mathbb N}} | \ \ \psi_k\in \dot{H}^s, \ \ \sum_{k\ge 1} \lambda_k \|\psi_k\|_{\dot{H}^s}^2<\infty\}$$ is a Banach space with norm
$\|\Psi\|_{\dot{\cH}^s} = (\sum_{k\ge 1}\lambda_{k}\|\psi_{k}\|_{\dot{H}^s}^{2})^{\frac{1}{2}}.$

\end{itemize}


\section{The model and main results}\label{sec:MainResults}

In this article, we study the global Cauchy problem for the semi-relativistic Schr\"odinger-Poisson system in the presence of linear dissipation. 
For $s\ge 0,$ we define the state space for the Schr\"odinger-Poisson system  by
$$
{\mathcal S^s}:=\{ (\Psi, \ulambda)  | \ \  \Psi=\{\psi_{k}\}_{k\ge 1}\in \cH^s \; \; is \; a \; complete \; 
orthonormal \; system
$$
$$
 \; in \; L^{2}({\mathbb R}^3), \ \ \ulambda = \{\lambda_k\}_{k\in{\mathbb N}}\in l^1, \; \lambda_k > 0\}.
$$
For $\Psi(t,x)=\{\psi_{k}(t,x)\}_{k\ge 1},$ the Schr\"odinger-Poisson system of equations, which is a system of infinitely coupled nonlinear equations, is given by
\begin{equation}
\label{eq:SP}
i\partial_t \Psi(t,x)=\sqrt{-\Delta + m^2}\Psi(t,x) +V[\Psi(t)](x)\Psi(t,x) -i \epsilon (-\Delta)^\alpha\Psi(t,x) , 
\end{equation}  
\begin{equation}
\label{eq:V}
\Delta V[\Psi(t)](x)= n[\Psi(t,x)], 
\end{equation} 
\begin{equation} 
\label{eq:n}
n[\Psi(t,x)]=\sum_{k=1}^{\infty}\lambda_{k}|\psi_{k}(t,x)|^{2} \, ,
\end{equation} 
where $t\ge 0,$ $x\in {\mathbb R}^3,$ $m,$ $\epsilon\ge 0,$ and $\alpha\ge \frac{1}{2}.$
Note that the nonlinear potential is focusing and is given by 
$$V[\Psi(t)](x) := (V[\Psi(t,\,\cdot\,)])(x) = - (\frac{1}{|\cdot|} \star n[\Psi(t,\,\cdot\,)])(x).$$
To explain the role of the coefficients $\{\lambda_k\}_{k\in {\mathbb N}}$, we note that the
Schr\"odinger-Poisson system is equivalent to the Heisenberg equations
$$i \partial_t \gamma = [\sqrt{-\Delta+m^2},\gamma]-[\frac1{|\cdot|}\star\rho,\gamma]-i\epsilon\{(-\Delta)^\alpha,\gamma\}$$
for the density matrices
$$\gamma(t,x,x') = \sum_k \lambda_k \psi_k(t,x) \overline{\psi_k}(t,x')$$
with initial data $\gamma_0$, provided that the sequence $\{\lambda\}_{k\in {\mathbb N}}$ is independent of $t$.
Here, $\rho(t,x):=\gamma(t,x,x)$, and $\{A,B\}:=AB+BA$ denotes the anticommutator.

We note that while our work is formulated for the state space ${\mathcal S^s}$ characterized
by infinite sequences $\{\lambda\}_{k\in {\mathbb N}}$,
where $\lambda_k>0$, one can easily accommodate  vanishing coefficients ($\lambda_k=0$ for some values of 
$k$) or
the finite rank case (where $\lambda_k=0$ for all $k$ larger than some $K\in{\mathbb N}$)
either through restriction to the subspace spanned by the nonvanishing coefficients, 
or by a density argument.

The following is our first main result about the global Cauchy problem in the presence of pseudo-relativistic diffusion.

\begin{theorem}\label{th:GWP}
Consider the system of equations (\ref{eq:SP})-(\ref{eq:n}) with focusing nonlinearity and with $\epsilon>0.$  
\begin{itemize}
\item[(i)] If $\alpha>\frac{1}{2}$ and $(\Psi(0),\ulambda)\in {\mathcal S}^0,$  then there is a unique mild solution $(\Psi,\ulambda)\in C([0,\infty], {\mathcal S}^0).$  
\item[(ii)] If $\alpha=\frac{1}{2}$ and $(\Psi(0),\ulambda)\in {\mathcal S}^\frac{1}{2},$  then there is a unique mild solution $(\Psi,\ulambda)\in C([0,\infty], {\mathcal S}^\frac{1}{2}).$  
\end{itemize}
\end{theorem}

We now discuss the blow up of the solution in the absence of dissipation. When $\epsilon=0,$ 
the energy functional associated with the semi-relativistic Schr\"odinger-Poisson system without dissipation, 
\begin{equation}
\label{eq:Energy}
\cE(\Psi) = \frac{1}{2}\langle \Psi, \sqrt{-\Delta + m^2}\Psi\rangle_{\cL^2} + \frac{1}{4}\langle \Psi, V[\Psi]\Psi\rangle_{\cL^2} ,
\end{equation}
is conserved, see Remark \ref{rm:ConservationEnergy}. We have the following result for blow up of a spherically symmetric solution.

\begin{theorem}\label{th:blow up}
Consider the system of equations (\ref{eq:SP})-(\ref{eq:n}) with focusing nonlinearity and in absence of dissipation, $\epsilon=0.$ Suppose that $(\Psi(0),\ulambda)\in {\mathcal S}^{\frac{1}{2}}$ such that $\Psi(0)=\{\psi_k(0)\}_{k\ge 1}$ is spherically symmetric, $\psi_k(0)\in C_c^\infty({\mathbb R}^3),$ and 
$$\cE(\Psi(0))<0.$$
Then there exists a finite positive time $\tau^*,$ such that 
$$\lim_{t\nearrow \tau^*}\|\Psi(t)\|_{\cH^{\frac{1}{2}}} = \infty, \ \ 0<\tau^*<\infty.$$
\end{theorem}


\section{Well-posedness in the presence of dissipation}\label{sec:Dissipation}

\subsection{Local well-posedness}\label{sec:LWP}

Note first that for $t\ge 0,$ $$e^{-it(\sqrt{-\Delta + m^2} -i\epsilon (-\Delta)^\alpha)},$$ with $\alpha,\epsilon,m\ge 0,$ is a contraction semigroup on $\cH^{s}, \ \ s\ge \frac{1}{2}.$ Furthermore, the nonlinearity is locally Lipschitz in $\cH^{s}, \ \ s\ge \frac{1}{2}.$ Using a standard contraction map argument, one can show local well-posedness of the Schr\"odinger-Poisson system of equations with linear dissipation.

\begin{proposition}\label{pr:LWP}
Consider the system of equations (\ref{eq:SP})-(\ref{eq:n}), with $\epsilon\ge 0.$ Suppose that $(\Psi(0),\underline{\lambda})\in {\mathcal S^s},$ $\frac{1}{2}\le s\le 2.$ Then there exists a positive time $T$ such that the unique solution $\Psi \in C([0,T]; \cH^s).$ Furthermore, there exists a maximal time $\tau^*\in (0,\infty]$ such that $$\lim_{t\nearrow\tau^*}\|\Psi(t)\|_{\cH^s}=\infty.$$ 
\end{proposition}

\begin{remark}\label{rm:ConservationCharge}
It follows from (\ref{eq:SP}) and local well-posedness that in the absence of dissipation ($\epsilon=0$), 
$\frac{d}{dt}\|\psi_k(t)\|_{L^2}^2   =  0,  \ \ k\ge 1.$
Therefore, $$\|\psi_k(t)\|_{L^2} = \|\psi_k(0)\|_{L^2}, \ \ 0\le t<\tau^*, \ \ k\ge 1$$
when $\epsilon=0.$
\end{remark}

\begin{remark}\label{rm:ConservationEnergy}
When $s=\frac{1}{2}$ and $\epsilon=0,$ the energy $\cE(\Psi(t))$ given by (\ref{eq:Energy}) is conserved. Formally, this follows from taking the $\cL^2$-inner product of $\partial_t\Psi(t)$ and right-hand-side of (\ref{eq:SP}). To make the argument precise, one applies a standard regularization argument using $(\delta\sqrt{-\Delta+m^2} + 1)^{-1}, \ \ \delta\ge 0,$ and functional calculus, see Lemma 2.4 in \cite{ACV13}.
\end{remark}

\begin{proof}
Note that the nonlinearity $V[\Psi]\Psi$ is locally Lipschitz:
For $\Psi,\Phi\in \cH^s, \ \ s\ge \frac{1}{2},$
\begin{equation}
\label{eq:Lipschitz}
\|V[\Psi]\Psi - V[\Phi]\Phi\|_{\cH^s} \lesssim (\|\Psi\|_{\cH^s}^2+\|\Phi\|_{\cH^s}^2) \|\Psi-\Phi\|_{\cH^s}, \ \ s\ge\frac{1}{2}.
\end{equation}
This inequality follows from Lemma 2.1 in \cite{ACV13}, where this property is proven for more general nonlinearities. For the sake of completeness, we briefly sketch the proof. 

From the Minkowski inequality, 
\begin{equation}
\|V[\Psi]\Psi - V[\Phi]\Phi\|_{\cH^s} \lesssim \|(V[\Psi] - V[\Phi])\Psi\|_{\cH^s} + \|V[\Phi](\Psi-\Phi)\|_{\cH^s} \label{eq:VDiff1}
\end{equation}
We begin by estimating the first term on the right.
\begin{align}
&\|(V[\Psi] - V[\Phi])\Psi\|_{\cH^s} \lesssim \sum_{k,l\ge 1} \lambda_k\lambda_l \|\frac{1}{|x|}\star (|\psi_l|^2 -|\phi_l|^2)\psi_k\|_{H^s} \nonumber\\
&\lesssim \sum_{k,l\ge 1}\lambda_k\lambda_l\{ \|\frac{1}{|x|}\star (|\psi_l|^2 -|\phi_l|^2)\|_{L^\infty}\|\psi_k\|_{H^s} + \|\frac{1}{|x|}\star (|\psi_l|^2 -|\phi_l|^2)\|_{W^{s,6}}\|\psi_k\|_{L^3} \}\nonumber\\
&\lesssim \sum_{k,l\ge 1}\lambda_k\lambda_l\{ \|\psi_l -\phi_l\|_{H^{\frac{1}{2}}}(\|\psi_l\|_{H^{\frac{1}{2}}}+\|\psi_l\|_{H^{\frac{1}{2}}})\|\psi_k\|_{H^s} + \||\psi_l|^2 -|\phi_l|^2\|_{L^{3}}\|\psi_k\|_{H^{\frac{1}{2}}} \}\nonumber\\
&\lesssim (\|\Psi\|_{\cH^s}^2+\|\Phi\|_{\cH^s}^2) \|\Psi-\Phi\|_{\cH^s}.\label{eq:VDiff2}
\end{align}
Here, we used the Minkowski inequality in the first line, fractional Leibniz rule (Lemma \ref{lm:fLeibniz} in the Appendix) in the second line, H\"older's inequality, fractional integration (Lemma \ref{lm:fIntegralOperator} in the Appendix), Lemma \ref{lm:fSmoothing} in the Appendix, and Sobolev embedding $H^{\frac{1}{2}}\hookrightarrow L^{3}.$ 
Similarly, 
\begin{align}
& \|V[\Phi](\Psi-\Phi)\|_{\cH^s} \lesssim \sum_{k,l\ge 1} \lambda_k\lambda_l \|\frac{1}{|x|}\star |\phi_l|^2(\psi_k-\phi_k)\|_{H^s} \nonumber \\
&\lesssim \|\Phi\|_{\cH^s}^2 \|\Psi-\Phi\|_{\cH^s} \label{eq:VDiff3}.
\end{align}
Now (\ref{eq:Lipschitz}) follows from inequalities (\ref{eq:VDiff1}), (\ref{eq:VDiff2}) and (\ref{eq:VDiff3}).


Given $\rho,T>0,$ consider the Banach space $${\mathcal B}^s_{T,\rho}=\{\Psi \in L^\infty_T (\cH^s): \; \|\Psi\|_{L^\infty_T\cH^s}\le\rho\}.$$ 
Let $U^{(m)}(t) = e^{-it(\sqrt{-\Delta+m^2}-i\epsilon(-\Delta)^\alpha)},$  $t\ge 0,$ which is a contraction semigroup. 
We define the mapping $\cN$  by
$${\cN}(\Psi)(t) = U^{(m)}(t) \Psi(0) - i\int_0^t U^{(m)}(t-t') V[\Psi(t')] \Psi(t')dt' ,$$
which is the solution for the semi-relativistic Schr\"odinger-Poisson system of equations given by the Duhamel formula. 
We show that $\cN$ is a mapping from $\cB^s_{T,\rho}$ into itself. 
\begin{align*}
\|\cN (\Psi)\|_{L^\infty_T \cH^s} &\le \|\Psi(0)\|_{\cH^s} + T \|V[\Psi]\Psi\|_{L^\infty_T \cH^s}\\
&\le \|\Psi(0)\|_{\cH^s} + T \sum_{k,l\ge 1} \lambda_k \lambda_l \|\frac{1}{|x|} \star |\psi_l|^2 \psi_k\|_{L^\infty_T H^s}\\
&\le  \|\Psi(0)\|_{\cH^s} + T \sum_{k,l\ge 1} \lambda_k \lambda_l \{\|\frac{1}{|x|}\star(|\psi_l|^2)\|_{L^\infty_TL^\infty} \|\psi_k\|_{L^\infty_T H^s} +  \\ & +  \|\frac{1}{|x|}\star(|\psi_l|^2)\|_{L^\infty_T W^{s,6}} \|\psi_k\|_{L^\infty_T L^3}\} .
\end{align*}
It follows that 
\begin{align*}
\|\cN (\Psi)\|_{L^\infty_T \cH^s} &\le \|\Psi(0)\|_{\cH^s} + T \sum_{k,l\ge 1} \lambda_k \lambda_l \{\|\psi_l\|^2_{L^\infty_TH^{\frac{1}{2}}} \|\psi_k\|_{L^\infty_T H^s} + \\ & +  \|\psi_l\|^2_{L^\infty_TL^3} \|\psi_k\|_{L^\infty_T H^s}\} \\
&\le \|\Psi(0)\|_{\cH^s} + T \sum_{k,l\ge 1} \lambda_k \lambda_l \{\|\psi_l\|^2_{L^\infty_TH^{\frac{1}{2}}} \|\psi_k\|_{L^\infty_T H^s}\}\\
&\le \|\Psi(0)\|_{\cH^s} + T (\sum_{l\ge 1} \lambda_l \{\|\psi_l\|^2_{L^\infty_TH^{\frac{1}{2}}} )(\sum_{k\ge 1}\lambda_k\|\psi_k\|^2_{L^\infty_T H^s})^{\frac{1}{2}}\\
&\le \|\Psi(0)\|_{\cH^s} + T \|\Psi\|_{L^\infty_T \cH^{\frac{1}{2}}}^2 \|\Psi\|_{L^\infty_T \cH^s}. 
\end{align*}
Given $\Psi(0)$, we choose $T$ and $\rho$ such that
\begin{equation*}
\|\Psi(0)\|_{\cH^s}\le \frac{\rho}{2}, \ \ T\rho^2 <\frac{1}{2}.
\end{equation*}
It follows from the last inequality and the Duhamel formula that  
$$\|\Psi\|_{L^\infty_T \cH^s} \le  \rho.$$ 

Since the nonlinearity is locally Lipschitz, $\cN$ is a contraction map for sufficiently small $T,$
\begin{align*}
&\|\cN(\Psi) - \cN(\Phi)\|_{L^\infty_T\cH^s} \le  T \|V[\Psi]\Psi - V[\Phi]\Phi\|_{L^\infty_T \cH^s}\\
&\lesssim  T\rho^2\|\Psi-\Phi\|_{L^\infty_T \cH^s}.
\end{align*}
Local well-posedness together with blow-up alternative follows from a standard contraction mapping argument, see for example, \cite{Cazenave96}.

\end{proof}

When $\epsilon>0$ and $\alpha>\frac{1}{2},$ we have a stronger result owing to time decay due the fractional dissipative operator $\epsilon (-\Delta)^{\alpha}.$  
\begin{proposition}\label{pr:LWP2}
Consider the system of equations (\ref{eq:SP})-(\ref{eq:n}), with $\epsilon\ge 0$ and $\alpha>0.$ Suppose that $(\Psi(0),\underline{\lambda})\in {\mathcal S^0}.$ Then there exists a positive time $T$ such that the unique solution $\Psi \in C([0,T]; \cL^2).$ Furthermore, there exists a maximal time $\tau^*\in (0,\infty]$ such that $$\lim_{t\nearrow\tau^*}\|\Psi(t)\|_{\cL^2}=\infty.$$ 
\end{proposition}

\begin{proof}
The proof relies on Strichartz type estimates for fractional Laplacian, see Lemma \ref{lm:Strichartz}, and a standard contraction mapping argument.

Let $U(t) = e^{-it\sqrt{-\Delta + m^2}} $ and $S_\alpha(t) = e^{-\epsilon t (-\Delta)^\alpha}.$ Note that $U(t)$ is a unitary operator, and $S_\alpha(t), \ \ t\ge 0,$ is contraction semigroup on $\cH^{s}, \ \ s\ge 0,$ that commutes with $U(t).$

We define the mapping $\cN$  by
$${\cN}(\Psi)(t) = U(t) S_\alpha(t) \Psi(0) - i U(t)\int_0^t S_\alpha(t-t') U(-t') V[\Psi(t')] \Psi(t')dt',$$
whose fixed point is the solution for the semi-relativistic Schr\"odinger-Poisson system of equations given by the Duhamel formula. 

Given $\rho,T>0,$ consider the Banach space $${\mathcal C}_{T,\rho}=\{\Psi \in L^\infty_T (\cL^2): \; \|\Psi\|_{L^\infty_T\cL^2}\le\rho\}.$$ 
We show that $\cN$ is a mapping from $\cC_{T,\rho}$ into itself. 
Since $U$ is unitary and communtes with $S_\alpha$ and $(1-\Delta)^{\frac{s}{2}},$
\begin{align}
\|\cN (\Psi)\|_{L^\infty_T \cL^2} &\le \|\int_0^tS_\alpha(t-t') U(-t')V[\Psi(t')]\Psi(t')dt'\|_{L^\infty_T \cL^2} \nonumber \\ & + \|S_\alpha(t)\Psi(0)\|_{L^\infty_T\cL^2} .\label{eq:LWPEst1}
\end{align}
It follows from Parseval's theorem and the fact that $e^{-\epsilon t |k|^{\frac{\alpha}{2}}} \le 1$ that
\begin{equation} 
\|S_\alpha(t)\Psi(0)\|_{L^\infty_T\cL^2} \le \|\Psi(0)\|_{\cL^2}.\label{eq:LWPEst2}
\end{equation}
Furthermore, it follows from the space-time estimate in Lemma \ref{lm:Strichartz} with $n=3,$$\nu=0,$ $r=p=2,$ $q=\infty,$ and $b=\frac{2}{3}$ that 
\begin{align*}
&\|\int_0^tS_\alpha(t-t') U(-t')V[\Psi(t')]\psi_l(t')dt'\|_{L^\infty_TL^2} \\
&\le \sum_{k\ge 1} \lambda_k \|\int_0^tS_\alpha(t-t') U(-t')\frac{1}{|x|}\star |\psi_k(t')|^2 \psi_l(t')\|_{L^\infty_T L^2}\\
&\le \sum_{k\ge 1} \lambda_k \int_0^t\|S_\alpha(t-t') U(-t')\frac{1}{|x|}\star |\psi_k(t')|^2 \psi_l(t')\|_{L^2}\\
& \le C T^{1-\frac{1}{2\alpha}}\sum_{k\ge 1}\lambda_k \|\frac{1}{|x|}\star |\psi_k|^2 \psi_l\|_{L^\infty_T L^{\frac{6}{5}}}\\
&\le  C T^{1-\frac{1}{2\alpha}}\sum_{k\ge 1}\lambda_k \|\frac{1}{|x|}\star |\psi_k|^2\|_{L^\infty_TL^3} \|\psi_l\|_{L^\infty_T L^{2}}\\
&\le C T^{1-\frac{1}{2\alpha}}\sum_{k\ge 1} \lambda_k \||\psi_k|^2\|_{L^\infty_TL^1} \|\psi_l\|_{L^\infty_T L^{2}}\\
&\le C T^{1-\frac{1}{2\alpha}}\sum_{k\ge 1}\lambda_k \|\psi_k\|_{L^\infty_TL^2}^2 \|\psi_l\|_{L^\infty_T L^{2}}\\
&\le C T^{1-\frac{1}{2\alpha}}\|\Psi\|_{L^\infty_T \cL^{2}}^2 \|\psi_l\|_{L^\infty_T L^{2}}.
\end{align*}
Here, we used the space-time estimate in the third line, H\"older inequality in the fourth line, and fractional integration (Lemma \ref{lm:fIntegralOperator}) in the fifth line.
It follows that 
\begin{equation}\label{eq:LWPEst3}
\|\int_0^tS_\alpha(t-t') U(-t')V[\Psi(t')]\Psi(t')dt'\|_{L^\infty_T\cL^2} \le C T^{1-\frac{1}{2\alpha}}\|\Psi\|_{L^\infty_T \cL^{2}}^3.
\end{equation}
Substituting estimates (\ref{eq:LWPEst2}) and (\ref{eq:LWPEst3}) in (\ref{eq:LWPEst1}) yields
\begin{equation}\label{eq:LWPEst4}
\|\cN (\Psi)\|_{L^\infty_T \cL^2} \le \|\Psi(0)\|_{ \cL^2} +C T^{1-\frac{1}{2\alpha}}\|\Psi\|_{L^\infty_T \cL^{2}}^3.
\end{equation}
With $\Psi(0)$ given, we choose $T$ and $\rho$ such that 
\begin{equation*}
\|\Psi(0)\|_{\cL^2}\le \frac{\rho}{2}, \ \ C T^{1-\frac{1}{2\alpha}} \rho^2 <\frac{1}{2},
\end{equation*}
when $\alpha>\frac{1}{2}.$
It follows from the last inequality and the Duhamel formula that  
$$\|\Psi\|_{L^\infty_T \cL^2} \le \rho, \ \ \alpha>\frac{1}{2}.$$ 
We now show that $\cN$ is a contraction mapping on $\cC_{T,\rho}$ for sufficiently small $T.$ 
\begin{align*}
&\|\cN(\Psi) - \cN(\Phi)\|_{L^\infty_T\cL^2}\\& \le \|\int_0^tS_\alpha(t-t') U(-t')\left(V[\Psi(t')]\Psi(t')-V[\Phi(t')]\Phi(t')\right)dt'\|_{L^\infty_T \cL^2}\\
&\le \|\int_0^tS_\alpha(t-t') U(-t')\left(V[\Psi(t')]-V[\Phi(t')]\right)\Psi(t')dt'\|_{L^\infty_T \cL^2}+\\& + \|\int_0^tS_\alpha(t-t') U(-t')V[\Phi(t')]\left(\Psi(t')-\Phi(t')\right)dt'\|_{L^\infty_T \cL^2}
\end{align*}
Using a Strichartz estimate and property of fractional integral operator as above, 
$$\|\cN(\Psi) - \cN(\Phi)\|_{L^\infty_T\cL^2} \le CT^{1-\frac{1}{2\alpha}}\rho^2 \|\Psi-\Phi\|_{L^\infty_T \cL^2}.$$
Local well-posedness together with blow-up alternative when $\alpha>\frac{1}{2}$ follows from a standard contraction mapping argument.

\end{proof}

\subsection{Global well-posedness}\label{sec:GWP}

We now show global well-posedness in the presence of pseudo-relativistic diffusion, i.e., $\epsilon>0$ in (\ref{eq:SP}).
We start by estimating the rate of change of the $\cL^2$ norm.

\begin{lemma}\label{lm:L2norm}
Suppose the hypotheses of Theorem \ref{th:GWP} hold. Then, for $t\in (0,\tau^*),$ where $\tau^*$ is  the maximal time appearing in Proposition \ref{pr:LWP}, 
$$
\|\Psi(t)\|_{\cL^2}^2 = \|\Psi(0)\|_{\cL^2}^2 - \epsilon \sum_{k\ge 1}\lambda_k \int_0^t \langle \psi_k(s),(-\Delta)^{\alpha}  \psi_k(s)\rangle ds. 
$$
In particular, $$\|\Psi(t)\|_{\cL^2}\le \|\Psi(0)\|_{\cL^2}$$
and
\begin{equation*}
\int_0^t \|\Psi(s)\|^2_{\dot{\cH}^{\frac{1}{2}}} ds \lesssim \frac{{C_{\Psi(0)}}}{\epsilon},
\end{equation*}
where $C_{\Psi(0)}$ is a positive constant that depends on $\|\Psi(0)\|_{\cL^2}.$
\begin{proof}
It follows from (\ref{eq:SP}) and Proposition \ref{pr:LWP}
\begin{align*}
\frac{d}{dt}\|\Psi(t)\|_{\cL^2}^2   &= \frac{d}{dt} \langle\Psi(t),\Psi(t)\rangle \\
&= -2\epsilon \sum_{k\ge 1}\lambda_k \langle \psi_k(t), (-\Delta)^{\alpha} \psi_k\rangle_{L^2} .
\end{align*}
Integrating both sides of the equality yields the claim of the lemma.
\end{proof}
\end{lemma}

This suffices to show global well-posedness in $\cL^2$ when $\alpha>\frac{1}{2}.$ When $\alpha=\frac{1}{2},$ we control the rate of growth of the $\cH^{\frac{1}{2}}$ norm of the solution in the presence of dissipation.

\begin{lemma}\label{lm:Hnorm}
Suppose the hypotheses of (ii) Theorem \ref{th:GWP} hold. Then, for $t\in (0,\tau^*),$ where $\tau^*$ is  the maximal time appearing in Proposition \ref{pr:LWP},
$$\|\Psi(t)\|_{\cH^{\frac{1}{2}}} \lesssim e^{\frac{C_{\Psi(0)}}{\epsilon}t},$$
where $C_{\Psi(0)}$ is a positive constant that depends $\|\Psi(0)\|_{\cL^2}$ and that is independent of $\epsilon.$
\end{lemma}
\begin{proof}
It follows from Lemma \ref{lm:L2norm} that 
\begin{equation}
\label{eq:L2Decrease}
\|\Psi(t)\|_{\cL^2}\le \|\Psi(0)\|_{\cL^2}
\end{equation}
and
\begin{equation}
\label{eq:IntHnorm}
\int_0^t \|\Psi(s)\|^2_{\dot{\cH}^{\frac{1}{2}}} ds \lesssim \frac{{C_{\Psi(0)}}}{\epsilon},
\end{equation}
where $C_{\Psi(0)}$ is a positive constant that depends on $\|\Psi(0)\|_{\cL^2}.$

For $\delta>0,$ it follows from (\ref{eq:SP}) that 
\begin{align*}
&\frac{d}{dt} \langle \Psi(t), \sqrt{-\Delta} (1-\delta \Delta)^{-\frac{1}{2}}\Psi(t)\rangle = -2\epsilon\{ \langle \Psi(t), \sqrt{-\Delta} (1-\delta \Delta)^{-\frac{1}{2}} \sqrt{-\Delta} \Psi(t)\rangle  \\& +2\Im \langle  \sqrt{-\Delta} (1-\delta \Delta)^{-\frac{1}{2}} \Psi(t), V[\Psi(t)]\Psi(t)\rangle \\
&\le2\epsilon\Im \langle  \sqrt{-\Delta} (1-\delta \Delta)^{-\frac{1}{2}} \Psi(t), V[\Psi(t)]\Psi(t)\rangle \\ 
&\le C\epsilon \|\Psi(t)\|_{\cH^{\frac{1}{2}}}^4,
\end{align*}
where the contant $C$ is independent of $\delta>0.$
Integrating both sides yields
\begin{align*}
& \langle \Psi(t), \sqrt{-\Delta} (1-\delta \Delta)^{-\frac{1}{2}}\Psi(t)\rangle -  \langle \Psi(0), \sqrt{-\Delta} (1-\delta \Delta)^{-\frac{1}{2}}\Psi(0)\rangle\\
&\le C\epsilon \int_0^t  \|\Psi(s)\|_{\cH^{\frac{1}{2}}}^4ds. 
\end{align*}
The right-hand-side is independent of $\delta.$ Taking $\delta\rightarrow 0$
\begin{align}
& \langle \Psi(t), \sqrt{-\Delta} \Psi(t)\rangle -  \langle \Psi(0), \sqrt{-\Delta}\Psi(0)\rangle \nonumber\\
&\le C \epsilon\int_0^t\|\Psi(s)\|_{\cH^{\frac{1}{2}}}^4ds. \label{eq:Hnorm2}
\end{align}
Since $$\sqrt{-\Delta}+1 \sim \sqrt{-\Delta+1},$$ it follows from the last inequality (\ref{eq:Hnorm2}) that
$$\|\Psi(t)\|_{\cH^{\frac{1}{2}}}^2 \le \|\Psi(0)\|_{\cH^{\frac{1}{2}}} +  C_{\Psi(0)} \int_{0}^t (1+\|\Psi(s)\|_{\dot{\cH}^{\frac{1}{2}}}^2)\|\Psi(s)\|_{\cH^{\frac{1}{2}}}^2ds, $$
where $C_{\Psi(0)}$ is a positive constant that depends on $\|\Psi(0)\|_{\cL^2}.$
Using Gronwall's Lemma, we have 
\begin{equation}
\label{eq:Hnorm3}
\|\Psi(t)\|_{\cH^{\frac{1}{2}}}^2 \le  \|\Psi(0)\|_{\cH^{\frac{1}{2}}}  e^{C_{\Psi(0)} \int_0^t \|\Psi(s)\|_{\dot{\cH}^{\frac{1}{2}}}^2 ds}.
\end{equation}
Now, (\ref{eq:IntHnorm}) and (\ref{eq:Hnorm3}) yield
$$
\|\Psi(t)\|_{\cH^{\frac{1}{2}}}^2 \le  \|\Psi(0)\|_{\cH^{\frac{1}{2}}}^2 e^{\frac{C_{\Psi(0)}}{\epsilon} t},
$$
for some constant $C_{\Psi(0)}$ that depends on  $\|\Psi(0)\|_{\cL^2}$ and that is independent of $\epsilon.$ 

\end{proof}

\begin{proof}[Proof of Theorem \ref{th:GWP}]

When $\alpha>\frac{1}{2},$ claim (i) follows from Proposition \ref{pr:LWP2} and Lemma \ref{lm:L2norm}.

When $\alpha=\frac{1}{2},$ claim (ii) follows from Proposition \ref{pr:LWP} and Lemma \ref{lm:Hnorm}.

\end{proof}


\section{Blow-up in the absence of dissipation}\label{sec:blow up}

In this section, we prove Theorem \ref{th:blow up} in the absence of dissipation ($\epsilon=0$). This is an extension of the analysis in \cite{FL07} to the case of infinitely coupled semi-relativistic NLS equations. Throughout this section, we assume that the hypotheses of Theorem \ref{th:blow up}.

\subsection{Propagation of Moments}
It follows from local well-posedness $\|\psi_k\|_{L^2}, \ \ k\in{\mathbb N},$ is a conserved quantity when $\epsilon=0.$
In this subsection, we show that under the hypotheses of Theorem \ref{th:blow up} $|x|^2\psi_k(t)\in L^2({\mathbb R}^3), \ \ t\in [0,\tau^*).$ 

\begin{lemma}\label{lm:Moments}
Consider the system of equations (\ref{eq:SP})-(\ref{eq:n}) with $\epsilon=0,$ such that $\psi_k(0)\in C_c^\infty ({\mathbb R}^3).$ Then 
$$|x|^j \psi_k(t)\in L^2({\mathbb R}^3), \ \ j=1,2, \ \ t\in [0,\tau^*).$$ 
\end{lemma}

\begin{proof}
Let $\varphi\in C_c^\infty({\mathbb R}^3)$ such that 
$$\varphi(x) = \begin{cases}
1, \ \ |x| \le 1\\
0, \ \ |x|\ge 2
\end{cases}.$$
For $k\in {\mathbb N}$ fixed, let 
$$f^{(j)}_\tepsilon(t) = \||x|^j \varphi(\tepsilon x) \psi_k(t)\|_{L^2}^2, \ \ \tepsilon\ge 0.$$ For $\tepsilon>0,$ this is a well-defined quantity for $t\in [0,\tau^*).$ 
For $t\in [0,\tau^*),$
\begin{align}
f^{(j)}_\tepsilon(t) &= f^{(j)}_\tepsilon(0) + \int_0^t \frac{d}{ds} f^{(j)}_\tepsilon (s) ds \nonumber\\
&=  f^{(j)}_\tepsilon(0) + \int_0^t \langle \psi_k(s), [\sqrt{-\Delta+m^2}, |x|^{2j} \varphi^2(\tepsilon x)] \psi_k(s)\rangle ds\nonumber\\
&\le  f^{(j)}_\tepsilon(0) + \int_0^t \|[\sqrt{-\Delta+m^2}, |x|^j \varphi(\tepsilon x)]\psi_k(s)\|_{L^2}  \sqrt{f^{(j)}_\tepsilon (s)}ds. \label{eq:fUB}
\end{align}

In the case $j=1,$ it follows from Lemma \ref{lm:CommutatorEstimate} and the fact that $$\|\nabla (|x|\varphi(\tepsilon x))\|_{L^\infty}<C_1,$$ with upper bound $C_1$ which is finite and independent of $\tepsilon,$ that 
$$f^{(1)}_\tepsilon(t) \le f^{(1)}_\tepsilon(0) + C_1\int_0^t\sqrt{f^{(1)}_\tepsilon (s)}ds.$$ Since $\psi_k(0)\in C_c^\infty,$ $f^{(1)}_\epsilon(0)\le C_2,$ where $C_2$ is finite and independent of $\tepsilon.$
It follows that 
$$\sqrt{f^{(1)}_\tepsilon(t)} \le C_1 t + C_2.$$ Taking the limit $\tepsilon\rightarrow 0$ yields the claim for $j=1.$

In the case $j=2,$ note that 
\begin{align*}
&\|[\sqrt{-\Delta+m^2}, |x|^2 \varphi(\tepsilon x)]\psi_k\|_{L^2} \\&= \|\left(\varphi(\tepsilon x) x\cdot [\sqrt{-\Delta+m^2},x]  + [\sqrt{-\Delta+m^2},x\varphi(\tepsilon x)] \cdot x\right) \psi_k\|_{L^2} \\
&\lesssim \|x\cdot\frac{p}{p^2+m^2} \psi_k\|_{L^2} + \||x|\psi_k\|_{L^2} ,
\end{align*}
where $p=-i\nabla$ is the momentum operator. In the last line, we used the fact that $[\sqrt{p^2+m^2},x]$ and $[\sqrt{p^2+m^2}, x\varphi(\tepsilon x)]$ are bounded operators on $L^2,$ which follows from the canonical commutation relation $[x,p] = i.$
To control the  weighted $L^2$ estimate for the singular operator $\frac{p}{\sqrt{p^2+m^2}},$ note that the weight $|x|^2$ belongs to class $A_2,$ see \cite{Stein93}, V.5. 
Therefore, 
$$\|[\sqrt{-\Delta+m^2}, |x|^2 \varphi(\tepsilon x)]\psi_k\|_{L^2} \lesssim \||x|\psi_k\|_{L^2},$$ where the implicit constant on the right-hand-side is independent of $\epsilon.$
Substituting in (\ref{eq:fUB}) with $j=2$ yields
$$\sqrt{f^{(2)}_\tepsilon(t)} \le C_1 t^2 + C_2 t+C_3,$$ where all the constants are independent of $\tepsilon.$ Taking the limit $\tepsilon\rightarrow 0$ yields the claim for $j=2.$
\end{proof}

\subsection{Variance-type estimate}

We define the operators $$M= x\sqrt{-\Delta + m^2} x = \sum_{j=1}^3 x_j \sqrt{-\Delta+m^2} x_j,$$ and the dilation operator $$A = \frac{1}{2}(x\cdot p +p\cdot x).$$
Note that the $M\ge 0.$
We have the following result that relates the rate of change of the expectation value of $M$ and the expectation value of $A.$

\begin{lemma}\label{lm:VarianceEstimate}
Suppose that the hypotheses of Theorem \ref{th:blow up} hold. Then, for $t\in (0,\tau^*),$
\begin{equation*}
\frac{d}{dt} \langle \Psi(t), M \Psi(t) \rangle_{\cL^2} \le 2 \langle \Psi(t), A \Psi(t) \rangle_{\cL^2} + C_{\Psi_0},
\end{equation*}
where $C_{\Psi_0}$ depends on $\|\Psi_0\|_{\cL^2}.$
\end{lemma}

\begin{proof}
The canonical commutation relation yields
\begin{align}
&[\sqrt{p^2+m^2},M] = -2i A \label{eq:HMComm}\\
&[V[\Psi],M] = [V[\Psi]x^2,\sqrt{p^2+m^2}] - i\frac{p}{\sqrt{p^2+m^2}}\cdot xV[\Psi] \label{eq:VMComm} \\ & \ \ \ \ \ \ \ \ \ \ - iV[\Psi]x\cdot \frac{p}{\sqrt{p^2+m^2}} \nonumber .
\end{align}

Since $|x|^2\psi_k(t) \in L^2({\mathbb R}^3),$ Lemma \ref{lm:Moments}, and $\Psi(t)\in \cH^2({\mathbb R}^3),$ it follows that
\begin{align*}
|\langle \Psi, M \sqrt{p^2+m^2} \Psi\rangle_{\cL^2}| &= |\sum_{k\ge} \lambda_k \langle \psi_k, M \sqrt{p^2+m^2}  \psi_k \rangle| \\
&\le \sup_{k\ge 1}\||x| \psi_k\|_{L^2} \|\Psi\|_{\cH^1} + \sup_{k\ge 1}\||x|^2\psi_k\|_{L^2} \|\Psi\|_{\cH^2}.
\end{align*}
Therefore, $$ \langle \Psi(t),i[ M, \sqrt{p^2+m^2}] \Psi(t)\rangle_{\cL^2}$$ 
is well-defined for $t\in (0,\tau*).$
It follows from Lemm \ref{lm:CommutatorEstimate} that 
$$\|[V[\Psi]x^2,\sqrt{p^2+m^2}]\|_{L^2\rightarrow L^2} \le C \|\nabla V[\Psi]x^2\|_{L^\infty}.$$

We now show that 
\begin{equation}
\label{eq:PotentialGradient}
\|\nabla V[\Psi (t)]x^2\|_{L^\infty} +\||x|V[\Psi (t)]\|_{L^\infty} \le C \|\Psi_0\|_{\cL^2}.
\end{equation}
Recall that
$$V[\Psi(t)](x) = -\sum_{k\ge 1} \lambda_k \int_{{\mathbb R}^3} \frac{1}{|x-y|} |\psi_k(t,y)|^2 dy,$$
where $\psi_k$ is spherically symmetric.
It follows from Newton's Theorem that 
$$|V[\Psi(t)](x)| \le \frac{1}{|x|} \|\Psi(t)\|_{\cL^2}^2 = \frac{1}{|x|} \|\Psi_0\|_{\cL^2}^2,$$
where we used conservation of charge (Remark \ref{rm:ConservationCharge}) in the last equality.
Using the explicit formula
$$V[\Psi(t)](x) = -\frac{1}{|x|} \sum_{k\ge 1}\lambda_k \int_{|y|\le r} |\psi_k(t,y)|^2 dy - \sum_{k\ge 1}\lambda_k \int_{|y|>r} \frac{|\psi_k(t,y)|^2}{|y|} dy,  $$
we have 
\begin{align*}
|\nabla V[\Psi(t)](x)| &= |\partial_r V[\Psi(t)](x)|  \le \frac{1}{|x|^2} \sum_{k\ge 1}\lambda_k\int_{|y|\le r} |\psi_k(t,y)|^2dy\\
& \le \frac{\|\Psi(t)\|_{\cL^2}^2}{|x|^2} = \frac{\|\Psi_0\|_{\cL^2}^2}{|x|^2} .
\end{align*}

It follows from (\ref{eq:SP}), (\ref{eq:HMComm}), (\ref{eq:VMComm}) and (\ref{eq:PotentialGradient}) that
\begin{align*} 
\frac{d}{dt} \langle \Psi(t), M \Psi(t) \rangle_{\cL^2} &=  \langle \Psi(t), i[\sqrt{p^2+m^2} + V[\Psi(t)], M] \Psi(t) \rangle_{\cL^2} \\
&= 2\langle \Psi(t), A \Psi(t) \rangle_{\cL^2} + \langle \Psi(t), i[V[\Psi(t)], M] \Psi(t) \rangle_{\cL^2} \\
&\le  2\langle \Psi(t), A \Psi(t) \rangle_{\cL^2} +C_{\Psi_0},
\end{align*}
where $C_{\Psi_0}$ depends on $\|\Psi_0\|_{\cL^2}.$
\end{proof}

\subsection{Dilation Estimate}

We now estimate the rate of change of $\langle \Psi(t), A \Psi(t) \rangle_{\cL^2}.$

\begin{lemma}\label{lm:dilationEstimate}
Suppose that the hypotheses of Theorem \ref{th:blow up} hold. Then, for $t\in (0,\tau^*),$ 
$$\frac{d}{dt}\langle \Psi(t), A \Psi(t) \rangle_{\cL^2} \le 2 \cE (\Psi_0). $$
\end{lemma}
\begin{proof}
The canonical commutation relation implies
\begin{equation}
[\sqrt{p^2+m^2} + V[\Psi], A] = -i\sqrt{p^2+m^2} + i\frac{m^2}{\sqrt{p^2+m^2}} + ix\cdot\nabla V[\Psi].\label{eq:HAComm}
\end{equation}
By interchanging differentiation and integration, 
\begin{equation}
\langle \Psi, x\cdot\nabla V[\Psi] \Psi \rangle_{\cL^2} = -\frac{1}{2}\langle \Psi, V[\Psi] \Psi \rangle_{\cL^2}. \label{eq:xGradV}
\end{equation}
It follows from (\ref{eq:SP}), (\ref{eq:HAComm}), (\ref{eq:xGradV}) and conservation of energy when $\epsilon=0$ (Remark \ref{rm:ConservationEnergy}) that 
$$\frac{d}{dt} \langle \Psi(t), A \Psi(t) \rangle_{\cL^2} = 2\cE(\Psi_0) - \langle \Psi(t), \frac{m^2}{\sqrt{p^2+m^2}} \Psi(t) \rangle_{\cL^2}\le 2\cE(\Psi_0).$$
\end{proof}

\subsection{Proof of Theorem \ref{th:blow up}}

\begin{proof}
Lemma \ref{lm:VarianceEstimate} and Lemma \ref{lm:dilationEstimate} imply that, for $t\in (0,\tau^*),$ 
$$\langle \Psi(t), M \Psi(t) \rangle_{\cL^2} \le 2\cE(\Psi_0) t^2 + C_1 t +C_2  $$
where the constants $C_1$ and $C_2$ depend on $\|\Psi_0\|_{\cL^2}.$
Since $M\ge 0$, while $\cE(\Psi_0)<0$ given, it follows that the maximal time of the solution is finite. The claim of the theorem follows from the blow-up alternative.
\end{proof}


\section*{Appendix}

For the benefit of the reader, we recall some useful results that are used in the analysis.

\subsection{Fractional differentiation and fractional integral operators}

The following result about the fractional Leibniz rule can be found in \cite{Kato95}.
\begin{lemma}\label{lm:fLeibniz}
$$\|\cD^s(uv)\|_{L^p} \lesssim \|\cD^s u \|_{L^{q_1}}\|v\|_{L^{r_1}} + \|u\|_{L^{q_2}}\|\cD^s v\|_{L^{r_2}},$$
where $\frac{1}{p} = \frac{1}{q_i}+\frac{1}{r_i}, \ \ i=1,2.$
\end{lemma}
 
The second result is about inequality involving fractional integral operators, which can be found, for example, in \cite{Stein93}. 
\begin{lemma}\label{lm:fIntegralOperator}
Let $I_\gamma$, for $0<\gamma<n$,
be the fractional integral operator 
$$I_\gamma (u) = \int_{\bbR^n} |x-y|^{\gamma-n} u(y)~dy.$$ Then 
$$\|I_\gamma (u)\|_{L^p} \lesssim \|u\|_{L^q}, \ \ \frac{1}{p} = \frac{1}{q}-\frac{\gamma}{n}.$$
\end{lemma}
In the analysis above, $n=3$ and $\alpha=2.$

We also recall the following useful Hardy-type inequality.

\begin{lemma}\label{lm:fSmoothing}
Let $0<\gamma<n$. Then,
$$\sup_{x\in {\mathbb R}^n} |\int_{{\mathbb R}^n} \frac{1}{|x-y|^\gamma} |u(y)|^2dy| \lesssim \|u\|_{\dot{H}^{\frac{\gamma}{2}}}^2 \,.$$
\end{lemma}

The following result follows from Calder\'on-Zygmund theory for singular operators, see \cite{Stein93}, VII.

\begin{lemma}\label{lm:CommutatorEstimate}
Suppose that $f$ is locally integrable and $\nabla f\in L^{\infty}({\mathbb R}^3)$ as a vector-valued function. Then, for $\delta\ge 0,$
$$\|[\sqrt{-\Delta+\delta},f]\|_{L^2\rightarrow L^2} \le C \|\nabla f\|_{L^\infty},$$
where $C$ is independent of $\delta.$
\end{lemma}

\subsection{Strichartz estimates for fractional Laplacian}

We finally recall results about Strichartz estimates for the fractional Laplacian, see \cite{MYZ08}, and also \cite{LZ10}. Let $S_\alpha(t) = e^{-t(-\Delta)^\alpha},$ where $t\ge 0,$ $\alpha\in (0,\infty)$ and $\Delta$ is the Laplacian in ${\mathbb R}^n.$ Then the kernel of $S_\alpha$ is $K_t(x) = {\mathcal F}^{-1}[e^{-t|k|^{2\alpha}}],$ where ${\mathcal F}^{-1}$ stands for the inverse Fourier transform. It is a bounded linear operator on $L^p({\mathbb R}^n).$  Furthermore, it follows from the scaling properties of the kernel and Young's inequality that, if $f\in L^r({\mathbb R}^n), \ \ r\ge 1,$ 
\begin{equation}\label{eq:TDecay}
\|(-\Delta)^{\frac{\nu}{2}}S_\alpha(t)f\|_{L^p({\mathbb R}^n)}\le C t^{-\frac{\nu}{2\alpha} - \frac{n}{2\alpha}(\frac{1}{r}-\frac{1}{p})}\|f\|_{L^r({\mathbb R}^n)},
\end{equation}
for $1\le r\le p\le \infty$ and $\nu\ge 0.$
The triplet $(q,p,r)$ is called an admissible triplet if 
$$\frac{1}{q} = \frac{n}{2\alpha} (\frac{1}{r} - \frac{1}{p}),$$
$$1<r\le p <\begin{cases}
\frac{nr}{n-2\alpha}, \ \ n>2\alpha\\
\infty, \ \ n\le 2\alpha
\end{cases}.
$$
We have the following space-time estimate, which is a slight generalization of a result in \cite{MYZ08} where $\nu=0$ to the case $\nu\ge 0.$ 
\begin{lemma}\label{lm:Strichartz}
Let $(q,p,r)$ be an admissible triplet. Let $b>0, \ \ T>0,$ $I=[0,T),$ $\nu\ge 0,$ and $r_0 = \frac{nb}{2\alpha}.$ Suppose that $p>b+1,$ $p<r(b+1),$ and $r\ge r_0>1.$ If $f\in L^{\frac{q}{b+1}}(I,L^{\frac{p}{b+1}}({\mathbb R}^n)),$ then 
\begin{align*}
&\|\int_0^t (-\Delta)^{\frac{\nu}{2}} S_\alpha (t-s) f(s,x) ds\|_{L^\infty (I,L^r({\mathbb R}^n))}\\ & \le C T^{1-\frac{nb}{2r\alpha} - \frac{\nu}{2\alpha}} \|f\|_{L^{\frac{q}{b+1}}(I,L^{\frac{p}{b+1}}({\mathbb R}^n))}. 
\end{align*}
\end{lemma}
\begin{remark}Since $S_\alpha$ commutes with $(1-\Delta)^\beta,$ $\|\cdot\|_{L^p({\mathbb R}^n)}$ in Lemma \ref{lm:Strichartz} can be replaced by $\|\cdot\|_{W^{\beta,p}({\mathbb R}^n)}.$
\end{remark}
\begin{proof}
It follows from (\ref{eq:TDecay}) that
\begin{align*}
&\|\int_0^t (-\Delta)^{\frac{\nu}{2}} S_\alpha (t-s) f(s,x) ds\|_{L^\infty (I,L^r({\mathbb R}^n))}\\ 
& \le C \int_0^t (t-s)^{-\frac{\nu}{2\alpha} - \frac{n}{2\alpha}(\frac{b+1}{p}-\frac{1}{r})} \|f(s)\|_{L^{\frac{p}{b+1}}}ds\\
&\le  C  \left(\int_0^t (t-s)^{-\frac{\nu}{2\alpha} - \frac{n}{2\alpha}(\frac{b+1}{p}-\frac{1}{r})m}ds\right)^{\frac{1}{m}} \|f\|_{L^{\frac{q}{b+1}}L^{\frac{p}{b+1}}},
\end{align*}
where $\frac{1}{m} = 1-\frac{b+1}{q}.$
Hence
\begin{align*}
&\|\int_0^t (-\Delta)^{\frac{\nu}{2}} S_\alpha (t-s) f(s,x) ds\|_{L^\infty L^r}\\ 
&\le C \left(\int_0^t (t-s)^{-1+m(1-\frac{\nu}{2\alpha} - \frac{nb}{2r\alpha})}ds \right)^\frac{1}{m} \|f\|_{L^{\frac{q}{b+1}}L^{\frac{p}{b+1}}}\\
&\le C T^{1-\frac{nb}{2\alpha r} - \frac{\nu}{2\alpha}} \|f\|_{L^{\frac{q}{b+1}}L^{\frac{p}{b+1}}} .
\end{align*}
\end{proof}


\section*{Acknowledgements}
WAS acknowledges the financial support of a Discovery grant from the Natural Sciences and Engineering Research Council of Canada. TC was supported by NSF grants DMS-1009448 and DMS-1151414 (CAREER).



\begin{thebibliography}{99}


\bibitem{ACV12}
W. ~Abou Salem, T. ~Chen, V. ~Vougalter.
On the well-posedness of the semi-relativistic Schr\"odinger-Poisson system.
Dyn. Partial Differ. Equ. {\bf 9} (2012), no. 2, 121-132.


\bibitem{ACV13-2}
W. ~Abou Salem, T. ~Chen, V. ~Vougalter.
Existence and nonlinear stability of stationary states for the semi-relativistic Schrodinger-Poisson system.
Ann. Henri Poincar\'e (2013), online first, DOI 10.1007/s00023-013-0270-8.


\bibitem{ACV13}
W. ~Abou Salem, T. ~Chen, V. ~Vougalter.
On the generalized semi-relativistic Schr\"odinger-Poisson system. 
Documenta Math. {\bf 18} (2013), 343-357.

\bibitem{A10}
I.~ Anapolitanos.
Rate of convergence towards the Hartree-von Neumann limit in the mean-field regime. 
Lett. Math. Phys. {\bf 98} (2010), 1-31.

\bibitem{AS10}
P. ~Antonelli, C. ~Sparber.
Global well-posedness for cubic NLS with nonlinear damping.
Comm. Partial Differential Equations {\bf 35} (2010), no. 12,  
2310--2328.

\bibitem{BMN10}
B. Baeumer, M. Meerschaert, M. Naber.
Stochastic models for relativistic diffusion.
Phys. Rev. E {\bf 82} (2010), 011132.

\bibitem{Cazenave96}
T. Cazenave.
\newblock {\em An Introduction to Nonlinear Schr\"odinger Equations.}
\newblock Textos de M\'etodos Matem\'aticos 26. Instituto de Matem\'atica, Rio de Janeiro, 1996.

\bibitem{CMS90}
R. Carmona, W. C. Masters, B. Simon. 
Relativistic Schr\"oedinger operators: asymptotic behavior of the eigenfunctions.
J. Funct. Anal. {\bf 91} (1990), 117-142.

\bibitem{CO07}
Y. Cho, T. Ozawa.
On the semi-relativistic Hartree-type equation.
SIAM J. Math. Anal. {\bf 38} (2007), no. 4, 1060--1074.

\bibitem{FK12}
G. ~Fibich, M. ~Klein.
Nonlinear-damping continuation of the nonlinear Schr\"odinger equation--a 
numerical study. Phys. D {\bf 241} (2012), no. 5, 519--527. 

\bibitem{FL07}
J. Fr\"ohlich, E. Lenzmann.
Blow up for nonlinear wave equations describing boson stars.
Commun. Pure Appl. Math.  {\bf 60} (2007), no.11, 1691--1705.

\bibitem{HLLS10}
C. Hainzl, E. Lenzmann, M. Lewin, B. Schlein
On blowup for time-dependent generalized Hartree-Fock equations.
Ann. Henri Poincare {\bf 11} (2010), no. 6, 1023--1052.

\bibitem{Kato95}
T. Kato. On nonlinear Schr\"odinger equations II.
 J. Anal. Math. {\bf 67} (1995), 281--306.

\bibitem{L07}
E. ~Lenzmann.
Well-posedness for semi-relativistic Hartree equations of critical type.
Math. Phys. Anal. Geom. {\bf 10} (2007), no.1, 43--64.

\bibitem{LZ10}
P. Li, Z. Zhai. Well-posedness and regularity of generalized Navier–Stokes equations in some critical Q-spaces. J. Func. Anal. {\bf 259} (2010), 2457--2519.



\bibitem{MYZ08}
C. Miao, B. Yuan, B. Zhang. Well-posedness of the Cauchy problem for the fractional power dissipative equations. Nonlinear Anal. {\bf 68} (2008), 461--484.

\bibitem{PSS05}
T. ~Passot, C. ~Sulem, P.L. ~Sulem.
Linear versus nonlinear dissipation for critical NLS equation. 
Phys. D {\bf 203} (2005), no. 3-4,  167--184.

\bibitem{R02}
M. Ryznar.
Green function for relativistic $\alpha$-stable process.
Potential Analysis {\bf 17} (2002), 1-23.

\bibitem{Stein93}
E. M. Stein. 
Harmonic Analysis.
Princeton University Press, New Jersey, 1993.



\end{thebibliography}
\end{document}